\documentclass[a4paper,USenglish,cleveref,autoref,thm-restate]{lipics-v2021}

\usepackage[algoruled,vlined,linesnumbered]{algorithm2e}
\usepackage{multicol}

\usepackage{tikz}
\usetikzlibrary{arrows.meta, positioning}

\newcommand*{\C}{\mathcal{C}}
\newcommand*{\D}{\mathcal{D}}
\newcommand{\Exp}{\mathbb{E}}
\newcommand*{\G}{\mathcal{G}}
\newcommand*{\nwspace}{\hspace*{.1em}} 
\newcommand*{\Otilde}{\widetilde{O}}
\newcommand*{\Prob}{\mathbb{P}}

\let\oldsqrt\sqrt
\def\hksqrt{\mathpalette\DHLhksqrt}
\def\DHLhksqrt#1#2{\setbox0=\hbox{$#1\oldsqrt{#2\,}$}\dimen0=\ht0
   \advance\dimen0-0.2\ht0
   \setbox2=\hbox{\vrule height\ht0 depth -\dimen0}%
   {\box0\lower0.4pt\box2}}
\renewcommand\sqrt\hksqrt

\renewcommand{\leq}{\leqslant}
\renewcommand{\geq}{\geqslant}
\renewcommand{\le}{\leqslant}
\renewcommand{\ge}{\geqslant}

\title{Simpler and Improved Replacement Path Coverings}

\author{Davide Bilò}{University of L'Aquila, Italy}{davide.bilo@univaq.it}{https://orcid.org/0000-0003-3169-4300}{}
\author{Shiri Chechik}{Tel Aviv University, Israel}{shiri.chechik@gmail.com}{}{This project received funding from the
European Research Council (ERC) under the European Union's Horizon
2020 Research and Innovation program,
grant agreement No.~803118
``The Power of Randomization in Uncertain Environments''.}
\author{Keerti Choudhary}{Indian Institute of Technology Delhi, India}{keerti@iitd.ac.in}{https://orcid.org/0000-0002-8289-5930}{The author is supported by the Indian Science and Engineering Research Board (SERB) under the Mathematical Research Impact-Centric Support (MATRICS) scheme, grant agreement No.~MTR/2025/001601.}
\author{Sarel Cohen}{Reichman University, Israel}{sarel.cohen@runi.ac.il}{https://orcid.org/0000-0002-8289-5930}{}
\author{Martin Schirneck}{Karlsruhe Institute of Technology, Germany}{martin.schirneck@kit.edu}{https://orcid.org/0000-0001-7086-5577}{The author is supported by the German Research Foundation (DFG),
grant agreement No.~556899211 ``Design, Analysis, and Engineering of Enumeration Algorithms''.}

\authorrunning{D.~Bilò, S.~Chechik, K.~Choudhary, S.~Cohen, and M.~Schirneck}

\Copyright{Davide Bilò, Shiri Chechik, Keerti Choudhary, Sarel Cohen, and Martin Schirneck}

\category{} 

\relatedversion{} 

\nolinenumbers 
\hideLIPIcs

\begin{document}

\maketitle

\begin{abstract}
	An important tool in the design of fault-tolerant graph data structures
	are $(L,f)$\emph{-replacement path coverings} (RPCs). 
	An RPC is a family $\mathcal{G}$ of subgraphs of a given graph $G$ such that,
	for every set $F$ of at most $f$ edges,
	there is a subfamily $\mathcal{G}_F \,{\subseteq}\, \mathcal{G}$ with the following properties.
	\begin{enumerate}
		\item No subgraph in $\mathcal{G}_F$ contains an edge of $F$.
		\item For each pair of vertices $s,t$ that have a shortest path in $G{-}F$ 
		with at most $L$ edges,\\ one such path also exists in some subgraph in $\mathcal{G}_F$.
	\end{enumerate}
	The \emph{covering value} of the RPC is the total number $|\mathcal{G}|$ of subgraphs.
	The \emph{query time} is the time needed to compute the subfamily $\mathcal{G}_F$ given the set $F$.

	Weimann and Yuster [TALG'13] devised a randomized RPC 
	with covering value $\widetilde{O}(fL^f)$ and query time $\widetilde{O}(f^2 L^f)$.
	This was derandomized by Karthik and Parter [TALG'24],
	who also reduced the query time to $\widetilde{O}(f^2 L)$.
	Their approach uses some heavy algebraic machinery involving error-correcting codes
	and an increased covering value of $O((cfL \log n)^{f+1})$ for some constant $c > 1$. 
	We instead devise a much simpler derandomization via conditional expectations that
	lowers the covering value back to $\widetilde{O}(fL^{f+o(1)})$ 
	and decreases the query time to $\widetilde{O}(f^{5/2}L^{o(1)})$,
	assuming $f = o(\log L)$.

	We also investigate the optimal covering value of any $(L,f)$-replacement path covering
	(deterministic or randomized) for different parameter ranges.
	We provide a new randomized construction as well as improving a known lower bound, 
	also by Karthik and Parter.
	For example, for $f = o(\log L)$, we give an RPC with $\widetilde{O}( (L/f)^f\, L^{o(1)})$ subgraphs
	and show that this is tight up to the $L^{o(1)}$ term.
\end{abstract}

\section{Introduction}
\label{sec:intro}

Graphs are powerful models in computer science representing various types of relationships 
between entities encountered in different applications.
While a wide range of algorithms and data structures has been developed for static graphs, 
where edges and vertices remain unchanged over the whole lifetime of the application,
real-world networks are often subject to failures. 
Many traditional graph algorithms need to recompute their solutions entirely 
when some component of the input changes,
even if it is only a small number edges.
In many practical cases, there is an a priori upper bound on the number of simultaneous failures.
Moreover, while it may be unpredictable where the faults occur, 
they are often transient due to an inherent repair mechanism in the network.
This motivates the \emph{fault-tolerant} or \emph{sensitivity} setting in data structure research.
There, a preprocessing algorithm is given the underlying graph without failures and a \emph{sensitivity} parameter~$f$.
Queries to the data structure then specify up to $f$ edges
and the task is to quickly report the properties of the graph in which the specified edges failed.
Over the last two decades, substantial advancements have been made in developing such \emph{sensitivity oracles}
for fundamental graph problems like connectivity~\cite{DuanP09a, DuanP10, DuanP:17, PatrascuT:07}, shortest paths~\cite{Bilo24ApproxDSOSubquadraticTheoretiCS, BiChChCoFrScFOCS24, ChCo20, ChCoFiKa17,  ChoS024, DeThChRa08, DeyGupta24, DuanP09a, DuRe22, GuRen21}, diameter and excentricity~\cite{Bilo23STDiameter,Bilo22Extremal,HenzingerL0W17}, and routing~\cite{CLPR12}.
Recently, similar ideas have also been applied to \textsf{NP}-hard problems, such as vertex cover,
$k$-path~\cite{AlmanHirsch22ExteriorAlgebras,Bilo22FixedParameterSensitivityOracles}, 
or $k$-clique~\cite{Bilo25IndexingSubnetworks}.

We are concerned with \emph{$f$-edge fault-tolerant distance sensitivity oracles} ($f$-DSOs),
which are sensitivity oracles for pairwise graph distances.
In more detail, an $f$-DSO for a graph $G = (V,E)$ is queried with triplets $(s,t,F)$ consisting of two vertices $s,t \in V$  and a set $F \subseteq E$ of at most $f$ edges.
The output of the oracle is the length $d(s,t,F)$ of the shortest path from $s$ to $t$
in the modified graph $G{-}F$.
This value $d(s,t,F)$ is called the \emph{replacement distance} and any shortest $s$-$t$-path in $G{-}F$ is a \emph{replacement path}.
Weimann and Yuster, in their seminal work~\cite{WY13}, presented a non-trivial $f$-DSO
supporting multiple edge failures. 
It has separate logic and data structures for so-called hop-short and 
hop-long replacement paths.
Here, a replacement path is \emph{hop-short} if it has at most $L$ edges,
where $L$ is some positive integer parameter.
For such paths, Weimann and Yuster gave an insightful construction that is now known as an 
$(L,f)$-\emph{replacement path covering} (RPC).\footnote{%
	The name was introduced later by Karthik and Parter in (the conference version of)~\cite{KarthikParter24DeterministicRPC_TALG}.
}
They defined a collection $\G$ of $\Otilde(f L^f)$ random subgraphs\footnote{%
	We use $n$ for the number of vertices of the input graph $G$, and $m$ for the number of its edges.
	For a positive function $g(m,n,L,f)$, we let $\Otilde(g)$ stand for $O(g \cdot \textsf{poylog}(n))$.
} of $G$ that with high probability\footnote{%
	\emph{With high probability} (w.h.p.)
	means with probability at least $1-n^{-c}$ for some constant $c > 0$.
} has the following property.
Whenever two vertices $s,t$ and a set $F$ of at most $f$ edges are such that
$s$ and $t$ indeed have a hop-short replacement path in $G{-}F$,
then there exists a subgraph $G^* \in \G$ such that $G^*$ contains no edge of $F$
and at least one such replacement path is retained in $G^*$.

The usefulness of RPCs for distance sensitivity oracles stems from the following observation.
Suppose we are given a query set $F$.
Scanning through $\G$ and filtering for those subgraphs that have no edge of $F$ gives a subfamily $\G_F$.
For vertices $s$ and $t$, let $\widehat{d}_F(s,t)$ be the minimum $s$-$t$-distance among all graphs in $\G_F$.
$G^* \in \G_F$ guarantees that $\widehat{d}_F(s,t)$ 
is equal to the \emph{true} replacement distance $d(s,t,F)$
whenever $s$ and $t$ have a replacement path with at most $L$ edges.
Even if $s$ and $t$ only have hop-long paths,
we still have $\widehat{d}_F(s,t) \ge d(s,t,F)$.
The computed value never underestimates the replacement distance
since any shortest path that contributed to $\widehat{d}_F(s,t)$
only uses edges from $G{-}F$.
A caveat of the construction in~\cite{WY13} is that,
in order to find $\G_F$, all graphs of the family $\G$ need to be scanned,
taking time $\Otilde(f^2 L^f)$.

Several applications of replacement path coverings
have since been explored in the context of DSOs~\cite{AlonChechikCohen19CombinatorialRP, Bilo24ApproxDSOSubquadraticTheoretiCS, ChCo20, GrandoniVWilliamsFasterRPandDSO_journal, KarthikParter24DeterministicRPC_TALG}, 
$k$-path sensitivity oracles~\cite{Bilo25IndexingSubnetworks}, 
fault-tolerant spanners~\cite{BraunschvigCPS15, DinitzK11, DR20}, and fault-tolerant strong-connectivity preservers~\cite{CC20}. 
Additionally, the concept has also been utilized in distributed computing~\cite{CPT20, HitronPater21BroadcastCONGEST, Par19a, PY19a, PY19b}.
In light of these applications, we adopt the following definition
that is equivalent to the one by Weimann and Yuster~\cite{WY13} but focuses on the subfamily $\G_F$
instead of the individual subgraph $G^*$.
The idea is that $\G_F$ provides good estimates for \emph{all} pairs $(s,t)$ simultaneously.
We have to formulate the definition carefully
so that it also applies to settings in which there are several shortest 
$s$-$t$-paths in $G{-}F$, some of which may have more than $L$ edges.


\begin{definition}[replacement path coverings, covering value, query time]
\label{def:RPC}
	Let $L$ and $f$ be positive integers and $G = (V, E)$ a graph.
	An $(L,f)$\emph{-replacement path covering} for $G$ is a family $\G$ of spanning subgraphs of $G$
	that has a subfamily $\G_F \,{\subseteq}\, \G$ for every set $F \,{\subseteq}\, E$ of $|F| \,{\leqslant}\, f$ edges
	such that the following two properties hold.
	\vspace*{.25em}
 	\begin{enumerate}
 		\item No subgraph in $\G_F$ contains an edge of $F$.
 		\vspace*{.25em}
 		\item For all $s,t \in V$ such that there exists a shortest path from $s$ to $t$ in $G{-}F$  
 			with at most $L$ edges, 
 			at least one subgraph in $\G_F$ also has such a path.
 	\end{enumerate}	
 	\vspace*{.25em}
 	The \emph{covering value} of the $(L, f)$-replacement path covering is the number of subgraphs in $\G$.
 	Its \emph{query time} is the time required to compute $\G_F$ from $F$. 
\end{definition}

Another interesting parameter of RPCs is the number $|\G_F|$ of subgraphs
relevant for a given query set $F$.
In down-stream applications,
this translates to the number of instances that need to be processed in the hop-short case
to find replacement distance $d(s,t,F)$.

Alon, Chechik, and Cohen~\cite{AlonChechikCohen19CombinatorialRP}
derandomized several distance sensitivity oracles
and asked whether also $(L, f)$-replacement path coverings in general can be derandomized.
This would provide a new deterministic tool for the design of sensitivity data structures
also for other applications beyond shortest paths.
Karthik and Parter~\cite{KarthikParter24DeterministicRPC_TALG} 
answered this question affirmatively using algebraic error-correcting codes.
Their RPC has a covering value of $O( (cfL \log_2 n)^{f+1})$ for some constant $c > 1$.
We abbreviate this to $\Otilde(fL)^{f+1}$
(the exponent outside of the parentheses is intentional).
This is more than the original number of $\Otilde(f L^f)$ subgraphs~\cite{WY13}.
In~\cite[Section~1.3]{KarthikParter24DeterministicRPC_TALG},
the authors explain that the increase is a direct consequence of the determinism of their construction.
However, the fact that their preprocessing is reproducable
is also the reason for their much better $\Otilde(f^2 L)$ query time.
Very recently, Bilò, Choudhary, Cohen, Friedrich, and Schirneck~\cite{Bilo25IndexingSubnetworks}
presented a new RPC for the parameter range $f = o(\log L)$
simultaneously reducing the covering value and query time (more details below).
Their construction is \emph{randomized}, which raises the question,
whether randomness is required to achieve this improved performance.

Our first result rejects this hypothesis.
We completely derandomize the result by Bilò et al.~\cite{Bilo25IndexingSubnetworks}
without any loss in the parameters.
We thereby obtain a deterministic $(L,f)$-replacement path covering
whose covering value is only slightly larger than the construction by Weimann and Yuster~\cite{WY13},
but achieves a query time that is sub-polynomial in $L$.
A summary of the related work and our own result can be found in \Cref{table:results_trees}.

\begin{table*}[t]
\caption{
Comparison of $(L,f)$-replacement path coverings for sensitivity $f = o(\log L)$.\\
Randomized results hold with high probability.
}
\vspace*{.5em}
\centering
\setlength{\tabcolsep}{5pt}
\renewcommand{\arraystretch}{1.25}
\resizebox{0.9\textwidth}{!}{
\begin{tabular}{@{}ccccc@{}}

\textbf{Covering Value} & \textbf{Query Time} & \textbf{Size of $\G_F$} & \textbf{Randomization} & \textbf{Reference} \\
\noalign{\hrule height 1pt}\\[-10pt]

$\Otilde(fL^f)$ & $\Otilde(f^2 L^f)$ & $\Otilde(f L^{f-|F|})$ & randomized &~\cite{WY13}\\[.25em]

$\Otilde(fL)^{f+1}$ & $\Otilde(f^2 L)$ & $\Otilde(fL)$ & deterministic &~\cite{KarthikParter24DeterministicRPC_TALG} \\[.25em]

$\Otilde(fL^{f + o(1)})$ & $\Otilde(f^{\frac{5}{2}}L^{o(1)})$ & $\Otilde(fL^{o(1)})$ & randomized &~\cite{Bilo25IndexingSubnetworks} \\[.25em]

$\Otilde(fL^{f + o(1)})$ & $\Otilde(f^{\frac{5}{2}}L^{o(1)})$ & $\Otilde(fL^{o(1)})$ & deterministic &~\autoref{thm:det-trees} \\[.25em]
  
$\Otilde\big(fe^f (\frac{L}{f})^{f+o(1)}\big)$ & $\Otilde\big(f^{\frac{5}{2}} \nwspace e^f (\frac{L}{f})^{o(1)}\big)$ & $\Otilde(f e^f)$ & randomized &~\autoref{thm:rand-trees}\\[.25em]
 
\end{tabular}
}
\label{table:results_trees}
\end{table*}

\begin{restatable}{theorem}{dettrees}
\label{thm:det-trees}
	Let $G$ be a graph (possibly directed and positively edge-weighted) with $n$ vertices.
	Let $f$ and $L$ be two positive integers, which may depend on $n$,
	such that $f = o(\log L)$.
	There exists a deterministic $(L,f)$-replacement path covering for $G$ 
	with covering value $\Otilde(fL^{f+o(1)})$
	and query time $\Otilde\big(f^{5/2}L^{o(1)}\big)$.
	The size of the computed subfamily $\G_F$ is $\Otilde(f L^{o(1)})$.
	\vspace*{.5em}
	
	\noindent
	For sensitivities up to $f = o(\log n)$, the covering value becomes $L^{f}n^{o(1)}$
	and the query time as well as the size of $\G_F$ increase to $n^{o(1)}$.
\end{restatable}

The derandomization by Karthik and Parter~\cite{KarthikParter24DeterministicRPC_TALG}
of the original RPC requires setting up the error-correcting codes and deriving from the so-called hit-and-miss hash functions.
We provide a much simpler algorithm to obtain a deterministic construction.
Conceptually, we show that RPCs are amenable to the method of conditional expectations,
avoiding heavy algebraic machinery.
This hopefully makes our approach more applicable also to other tools 
for building fault-tolerant data structures.
A disadvantage of our technique is that it requires small sensitivity.
On the one hand, the best covering value and query time are achieved for $f = o(\log L)$.
On the other hand, the derandomization itself requires access to the pre-computed 
answers to all $O(n^2 \nwspace m^f)$ queries.
(See~\autoref{sec:overview} for more details.)

Besides the derandomization, another contribution of the work by Karthik and Parter~\cite{KarthikParter24DeterministicRPC_TALG}
is a lower bound on the covering value of any $(L,f)$-replacement path covering. 
They showed that whenever $(L/f)^{f} \le n$, 
there exists an $n$-vertex graph for which the family $\G$
must contain $\Omega(\nwspace (L/f)^f)$ subgraphs.
This leaves an $\Otilde(f^{f+1})$ gap to the covering value in~\cite{WY13}.
We aim at narrowing this gap, starting with a new lower bound.

\begin{restatable}{theorem}{lowerbound}
\label{thm:lower-bound}
For all positive integers $n, f, L$, with $L \ge 2$, there is an $n$-vertex graph $G$
s.t.\ any $(L, f)$-replacement path covering for $G$ has covering value at least
$\min\!\left\{\sum_{i=0}^{f-1} \binom{L-2}{i}, n \right\}$.
\end{restatable}

\noindent
Since the sum of binomial coefficients can be a bit unwieldy,
we provide closed forms for certain ranges of $f$ and $L$.

\begin{restatable}{corollary}{corlb} 
\label{cor:lb}
	Let the notation be the same as in~\autoref{thm:lower-bound}.
	\vspace*{.25em}
	\begin{enumerate}
		\item If there exists a constant $\varepsilon > 0$ such that
			 $2 \le f \le (1{-}\varepsilon)\nwspace \frac{L}{2}$, 
			then any $(L,f)$-replacement\\[.25em]
			path covering must have covering value 
			$\Omega\!\left(\min\!\left\{\sqrt{f \nwspace e^f}
				\, \frac{L^{f-1}}{f^f},\ n \right\} \right)$.
		\vspace*{.25em}
		\item If $f \ge \frac{L}{2} $, any $(L,f)$-replacement path covering
			must have covering value $\Omega\!\left(\min\!\left\{2^L, n \right\} \right)$.
	\end{enumerate}
\end{restatable}
\vspace*{.5em}

Compared to the lower bound in~\cite{KarthikParter24DeterministicRPC_TALG}, 
the one in~\autoref{cor:lb} (\emph{i}) trades a factor $L$ for $\sqrt{f \nwspace e^f}$,
which is asymptotically larger whenever $f \ge 2 \ln L$.
To also tackle the gap for $f = o(\log L)$, we improve the upper bound instead.
Recall that in this range the $\Otilde(fL^f)$ covering value by Weimann and Yuster~\cite{WY13} 
is still the best known.
We decrease this by a factor $(f/e)^f$ by giving a tighter analysis of the
construction algorithm of the (randomized)
sampling trees by Bilò et al.~\cite{Bilo25IndexingSubnetworks}
which also underpinned our derandomization.

\pagebreak

\begin{theorem}
\label{thm:rand-trees}
	Let $G$ be a graph (possibly directed and positively edge-weighted) with $n$ vertices.
	Let $f$ and $L$ be two positive integers, which may depend on $n$, such that $f = o(\log L)$.
	There exists a randomized $(L,f)$-replacement path covering for $G$ 
	that with high probability has
	covering value $\Otilde\big(f e^f \nwspace (\frac{L}{f})^{f+o(1)}\big)$
	and query time $\Otilde\big(f^{\frac{5}{2}}  e^f (\frac{L}{f})^{o(1)}\big)$.
	The size of the computed subfamily $\G_F$ is $\Otilde(fe^f)$.
	The  $(L,f)$-replacement path covering also supports vertex failures.
	\vspace*{.5em}
	
	\noindent
	For sensitivities up to $f = o(\log n)$, the covering value becomes $(\frac{L}{f})^f n^{o(1)}$,
	the query time and size of $\G_F$ is $n^{o(1)}$.
\end{theorem}

\autoref{table:lower_bounds} summarizes the new and known upper and lower bounds.
In the range $f = o(\log L)$, we shrink the gap between the bounds to
$\Otilde(f^{1-o(1)} e^f L^{o(1)}) = \Otilde( L^{o(1)})$.
For $\Omega(\log L) \le f \le (1{-}c) \frac{L}{2}$, that gap 
is now also smaller than before,
but currently remains at $\Otilde(f^{f+\frac{1}{2}} e^{-\frac{f}{2}} L)$.
In the general literature on distance sensitivity oracles
for non-constant values of $f$,
the sensitivity is expressed in terms of $n$ 
(while $L = L(n)$ is only an internal parameter).
A common setting is $f = o(\log n/\log\log n)$, see e.g.~\cite{BiChChCoFrScFOCS24,ChCoFiKa17,WY13}.
\Cref{thm:det-trees,thm:rand-trees} also apply to the larger
range of $f = o(\log n)$.
There, we get a gap of $n^{o(1)} L/\sqrt{f e^f}$.
If additionally $f \ge 2 \ln L$, then this becomes $n^{o(1)}$,
showing that our construction is near-optimal.

We conjecture that the true covering value is of order
$\Theta(e^f (\frac{L}{f})^f) \cdot \textsf{poly}(f)$,
whenever this is not larger than $n$.
To show this, new constructions are needed for both the upper and lower bound.
As a possibly more accessible open question,
we ask whether one can achieve a \emph{deterministic} $(L,f)$-replacement path covering
with the same parameters as in~\autoref{thm:rand-trees}.
\vspace*{.5em}

\noindent
\textbf{Outline.}
The remainder of the paper is structured as follows. 
In~\autoref{sec:overview},
we provide an overview of the sampling tree construction and our derandomization, 
which is then carried out in~\autoref{sec:det-RPC}. 
In~\autoref{sec:rand-RPC}, we present 
our new  randomized RPC. 
We prove the lower bound on the covering value in~\autoref{sec:lower}.

\begin{table*}[t]
\centering
\renewcommand{\arraystretch}{1.25}
\caption{
Upper and lower bounds on the covering value of $(L,f)$-replacement path coverings.
With $\varepsilon > 0$ we denote an arbitrarily small constant. 
}
\vspace*{.5em}
\begin{tabular}{ccc}

\textbf{Covering Value}  & \textbf{Parameter Range} & \textbf{Reference} \\
\noalign{\hrule height 1pt}\\[-14pt]

   $\Otilde\big((\frac{L}{f})^f L^{o(1)}\big)$ & $f = o(\log L)$  & \autoref{thm:rand-trees}\\[.5em]
  $(\frac{L}{f})^{f} n^{o(1)}$ & $f = o(\log n)$ & \autoref{thm:rand-trees}\\[.5em]

  $\Omega\big((\frac{L}{f})^{f}\big)$ & & \cite{KarthikParter24DeterministicRPC_TALG}\\[.25em]
  	
  \noalign{\hrule height 1pt}\\[-14pt]
  
  $\Otilde(f L^f)$ & $2 \ln L \le f \le (1{-}\varepsilon)\frac{L}{2}$ & \cite{WY13}\\[.5em]

  $\Omega\big(\sqrt{f \nwspace e^f} \, \frac{L^{f-1}}{f^f} \big)$ & & \autoref{cor:lb}.1\\[.25em]

  \noalign{\hrule height 1pt}\\[-14pt]

  $\Omega(2^L)$ & $f = \Omega(\log n); \frac{L}{2} \le f$ & \autoref{cor:lb}.2
\end{tabular}
\label{table:lower_bounds}
\end{table*}

\section{Overview} 
\label{sec:overview}

Our $(L,f)$-replacement path covering and 
the associated query data structure leverage a hierarchical sampling strategy 
by Bilò et al.~\cite{Bilo25IndexingSubnetworks}, which in turn generalized earlier work 
by Weimann and Yuster~\cite{WY13}.
 We briefly review the randomized construction 
before discussing the challenges of its derandomization and how we overcome them.
Throughout this overview, unless stated otherwise, we assume $f = o(\log L)$
to simplify the exposition.

\pagebreak

\noindent
\textbf{The sampling tree framework.}
The randomized RPC consists of a forest of sampling trees that are constructed in levels. 
Each tree has height $h$ and branching factor $\alpha = L^{f/h}$. 
Its nodes represent subgraphs of $G$ starting with the edge-less graph in the root.
A child node inherits all edges from its parent and re-inserts additional edges from $G$ 
independently with probability $1-p = 1-L^{-1/h}$. 
After $h$ levels, in each subgraph that is stored in a leaf,
any edge of $G$ is present with probability $1- L^{-1}$, 
matching the original model~\cite{WY13}.
However, the critical difference is
that the leaves are \emph{not} independent.
Any two leaves share at least all edges of their lowest common ancestor.

This locality is essential for query efficiency. 
The query algorithm performs a depth-first search 
starting from the root and, in each step,
is traversing to the first child that contains no edge of the given query set $F$.
This runs in time $O(f h \alpha) = O(h L^{(f/h) + o(1)})$, 
which is dramatically faster than examining all $\Otilde(L^f)$ subgraphs as in~\cite{WY13}. 
However, the straight-forward search may fail to reach a
suitable leaf in any given tree with probability $1-C^h$ for some constant $C > 1$. 
By repeating this search in $K = \Otilde(f C^h)$ independent trees,
the RPC achieves high success probability over all 
$n^2 \binom{m}{\le f} =  O(n^2 \nwspace m^f)$ queries $(s,t,F)$.
\vspace*{.5em}

\noindent
\textbf{Challenges of the derandomization.}
At first glance, derandomizing this construction via the method of conditional expectations 
appears straightforward: replace each random decision to include an edge in any of the subgraphs 
with a greedy choice that maximizes conditional expectations.
Unfortunately, this approach fundamentally breaks the query structure.
The randomized construction succeeds because of the following two reasons.
\vspace*{.25em}
\begin{enumerate}
	\item The failures are independent across trees, allowing boosting via repetition.
	\vspace*{.25em}
	\item The probability for the search to fail in a fixed tree 
		depends only on the height $h$.
\end{enumerate}

The greedy derandomization, basing decisions only on local information, 
runs the danger of correlating failures across trees. 
Even worse, the search looks for \emph{any} child whose subgraph has no edge of $F$.
This means that if we bias the tree structure to make some children succeed more often,
we may cause the failures to concentrate in other parts of the tree,
ruining the balance needed for the query algorithm.

Consider the collection $\C$ of pairs $(P,F)$
where $F \subseteq E$ is a set of at most $f$ failing edges and $P \subseteq E$ 
are the edges of a shortest path in $G{-}F$ with hop-length not larger than $L$.
The set $\C$ may contain a pair for each possible query.
Ideally, the derandomization procedure should maximize the number of such pairs
that will \emph{eventually} be covered by the leaves of \emph{any} of the trees. 
This is not the same as maximizing the likelihood that the depth-first search 
finds \emph{some} suitable leaf for a given pair.
The local greedy optimization and the global query structure are misaligned.
\vspace*{.5em}

\noindent
\textbf{Our derandomization technique.}
We resolve this through a novel derandomization framework that carefully balances 
the competing objectives.
Let $x$ be a node in one of the trees and $A_x \subseteq E$ those edges that are \emph{missing}
in the subgraph associated to $x$.
If $x$ is a leaf,
it covers a pair $(P,F) \in \C$ if and only if $F \subseteq A_x$ and $P \subseteq E{\setminus}A_x$.
As mentioned above, maximizing those pairs for which all failing edges are removed 
($F \subseteq A_x$) might inadvertently increase the number of pairs
for which $P \cap A_x$ is too large.
This is reflected in our definition of well-separated and poorly-separated pairs.
\vspace*{.5em}
\begin{itemize}
	\item \emph{Well-separated pairs} are those for which 
		the failing edges $F$ are already removed at node $x$,
		but the replacement path $P$ remains largely intact meaning that $|P \cap A_x|$ is small.
	\vspace*{.5em}
	\item \emph{Poorly-separated pairs} are those for which the removal has already gone too far, 
		damaging the replacement path beyond recovery at the current depth of $x$.
\end{itemize}

Let $w$ be the parent of the current node $x$.
We have to decide which edges to re-insert in the associated subgraph 
in addition to the ones inherited from $w$.
That means, we construct the set $A_x$ as a subset of $A_w$.
We can restrict our attention to those pairs $(P,F)$ that have $F \subseteq A_w$
and for which $F \nsubseteq A_{w'}$ holds for all siblings $w'$ of $x$
that are processed before $x$ by the deterministic query algorithm.
Let this be the set $\D_x$.
It contains much fewer pairs than the full collection $\C$.
Let further $X_{P,F,x}$ be the indicator variable whether the pair $(P,F) \in \D_x$ is well separated
w.r.t.\ $x$ and define $Y_{P,F,x}$ analogously for poorly-separated pairs.
The latter pairs
cannot be handled by the current tree and must be deferred to a later one.
However, we do not want to increase the total number of trees beyond our budget $K$. 
So rather than greedily maximizing well-separated pairs alone, exacerbating the imbalance,
we optimize the expectation of
\begin{equation*}
\sum_{(P,F) \in \D_x} \left( X_{P,F,x} - \frac{1}{2} \nwspace Y_{P,F,x} \right).
\end{equation*}
The expectation is taken over the original random sampling process
conditioned on all previous decisions taken by the derandomization.
The coefficient $1/2$ in the sum serves as a tie breaker.
Even if the number of well and poorly-separated pairs
were equal, we have
$\Exp[ \nwspace \sum_{(P,F)} (X_{P,F,x} - \tfrac{1}{2} Y_{P,F,x}) \nwspace] = \frac{1}{2} \nwspace \Exp[\nwspace \sum_{(P,F)} X_{P,F,w} \nwspace ] > 0$ for the expectation.
Thus, we do not lose significant ground for the well-separated pairs.

We prove that this weighted objective maintains the properties that are critical for the query algorithm.
The worst-case failure behavior across trees remains approximately independent,
even in the now fully deterministic construction. 
Also, by carefully choosing the depth-dependent threshold of what counts as ``small''\footnote{%
	The precise meaning of small turns out to be $|P \cap A_x| \le p^d L$,
	where $d$ is the depth of $x$ in the tree
	and $p$ is such that $p^f  = 1/\alpha$ for the branching factor $\alpha$.
		}
for $|P \cap A_x|$,
we retain the same success guarantee for an individual tree as in the randomized case.
\vspace*{.5em}

\noindent
\textbf{Improving the covering value.}
As a secondary contribution beyond the derandomization,
we also lower the upper bound on number of subgraphs.
We observe that the prior randomized construction in \cite{Bilo25IndexingSubnetworks} was not optimized 
with respect to the specific structure of the query algorithm. 
The fixed-order search through the children of the current node
gets stuck only when \textit{all} children contain some edge of $F$. 
This observation allows us to use lower branching factors $\alpha$ and smaller height $h$,
resulting in fewer leaves per tree.

However, these changes also affect the success probability in other ways.
For example, fewer child nodes being available in each transition
increase the failure probability.
We provide a sharper analysis of the  hierarchical sampling process of the trees
to prove that the reduction of $\alpha$ and $h$
can be counter-acted by a lower re-insertion probability $1-p$
and a slightly larger number of trees $K$,
while still maintaining the properties of an RPC with high probability.
This improves the covering value of the sampling tree framework by a factor $f^f$
from $\Otilde(L^{f+o(1)})$ to $\Otilde((L/f)^{f+o(1)})$.
\vspace*{.5em}

\noindent
\textbf{A new lower bound.}
Another contribution is an improved lower bound on the covering value 
of any $(L,f)$-replacement path covering (deterministic or randomized). 
Prior to our work, the only lower bound was $\Omega(\min\{ (L/f), \nwspace n\})$ by Karthik and Parter~\cite{KarthikParter24DeterministicRPC_TALG}.
It already shows that our $\Otilde((L/f)^{f+o(1)})$ construction for $f = o(\log L)$ is near-optimal.
However, the bound leaves a gap for larger sensitivities. 
We improve it significantly by constructing a family of weighted directed acyclic graphs 
that demonstrates that 
$\Omega(\min\{\sqrt{f e^f} \nwspace L^{f-1}/f^f,\nwspace n\})$ 
subgraphs are necessary.

At a high level, it consists of a directed binary tree $T$ (not related to the sampling trees).
It is constructed from what we call \emph{inner trees}.
These are comb-like structures consisting of a rooted path in which each node has an attached leaf.
The edge weights are set up in such a way that traveling along the back of the comb is for free
while the leaf that is farthest from the root in hop-distance is actually the closest
in weighted distance.
The inner trees are iteratively combined to form the binary tree $T$
by replacing the leaves of an inner tree
with new appropriately scaled-down inner trees.
The scaling is such that after $f$ rounds,
all leaves of the whole tree $T$ have the same hop-distance $L$ from its root.
The leaves of $T$ are all connected to a single sink node $v$. 

For each such leaf $x$, we define a failure set $F_x$ of at most $f$ edges
so that removing those edges makes $x$ the leaf that
is closest to the root in weighted distance.
In other words, the unique shortest replacement path in $G{-}F_x$ 
from the root to the sink $v$ goes through $x$.
This creates a barrier for small $(L,f)$-replacement path coverings 
since any RPC for $G$ must contain a different subgraph for every leaf.
A combinatorial argument then shows that there are at least
 $\sum_{i=0}^{f-1} \binom{L-2}{i}$ such leaves.

\section{Deterministic Replacement Path Coverings} 
\label{sec:det-RPC}

Let $\C$ be a collection of pairs $(P,F)$ of disjoint sets of edges,
satisfying $|P|\leq L$ and $|F|\le f$.
Here, $P$ corresponds to a shortest path between a source-target pair in $G{-}F$.
Our goal is to cover each pair $(P,F) \in \C$,
meaning, in at least one tree we construct, there is a leaf $x$ such that
the associated subgraph $G_x$ contains no edge of $F$ but all edges of $P$.
Moreover, the query algorithm that starts in the root of that tree and always 
recurses into the first child that does not have an edge of $F$ reaches $x$.
The remaining notation follows~\cite{Bilo25IndexingSubnetworks}.
We remark, however, that we slightly adjust the parameters to our needs.
(They thus also differ from the values reported in the first part of the overview in~\autoref{sec:overview}.)
Nevertheless, this will lead to the same asymptotic covering value and query time.

\subsection{Preliminaries}
\label{sec:det-RPC_prelims}

For a positive integer $\ell$, we use $[\ell]$ to denote the set $\{1, 2, \dots, \ell\}$.
Our data structure consists of $K$ disjoint rooted trees $\{T_i\}_{i \in [K]}$, 
whose nodes all represent a spanning subgraph of $G$. 
Each sampling tree has height $h$ and any internal node has exactly $\alpha$ children.
The parameters $K$, $h$, and $\alpha$ will be optimized later.
A single tree has $\alpha^h$ leaves and $O(\alpha^h)$ nodes in total.
We associate with each node $x$ in a tree $T_i$ a set $A_x \subseteq E$.
Namely, $x$ represents the graph $G_x = G{-}A_x$.
The sets $A_x$ is computed deterministically.
However, let us first consider the following random construction. 
If $x$ is a root, we simply set $A_x = E$.
Now let $y$ be a child of $x$, its set $A_y \subseteq A_x$ is obtained by selecting each edge in $A_x$ independently with probability $p = \alpha^{-1/f}$.
For our derandomization, it will also be important to adjust the parameters
	such that $p^hL < 1$ holds.
The random construction is iterated until height $h$. 
All random choices are made independently.
The total number of stored subgraphs is $O(K \alpha^h)$ and the RPC is given by the family $\G$ of all subgraphs that are stored in the leaves of all the $K$ sampling trees. 
The following lemma summarizes the structural properties of the randomized construction.

\begin{lemma}[Lemma~6 in~\cite{Bilo25IndexingSubnetworks}]
\label{lem:tree-property}
	Let $i \in [K]$ and $x$ be a node of the tree $T_i$ at depth $d$.
	For any edge $e \in E$, the probability of it being included in $A_x$ is
	$\Prob[e\in A_x]=p^d$.
	For any other edge $e' \in E$, $e' \neq e$, 
	the events $[e \in A_x]$ and $[e' \in A_x]$ are independent.
\end{lemma}

\begin{algorithm}[t]
\caption{Query algorithm on input $F$.}
\label{alg:tree-exact}
	$\G_F \gets \emptyset$\;
	\For{$i=1$ \KwTo $K$}
	{
		$x \gets$ root of $T_i$\;
		\While{$x$ \textup{is not a leaf}}
		{

  	     Let $y$ be the first child of $x$ satisfying $F\subseteq A_{y}$.
  	     If no such child exists,\\ then {\em continue} to outer for-loop\;
    	 $x\gets y$\;
		}
		\tcc{Node $x$ is a leaf at depth $h$ such that $F \subseteq A_x$.}
		$\G_F \gets \G_F \cup \{G_x\}$\; \label{line:result}
	}
\end{algorithm}

Given a failure set $F\subseteq E$ with $|F|\leq f$, 
\autoref{alg:tree-exact} computes the subfamily $\G_F$.
Recall that the set $A_x \subseteq E$ contains those edges that are \emph{removed} from the subgraph $G_x$ of the node $x$.  The trees $T_1, \dots, T_{K}$ are searched individually, starting in the respective roots. In each step, the algorithm always chooses the \emph{first} child node $y$ of the current node $x$ that satisfies $F \subseteq A_y$.
If no such child exists, the next tree is searched.
If eventually a leaf is reached, the subgraph stored there is added to $\G_F$. 
Observe that no graph in $\G_F$ contains a failing edge from $F$
since it is explicitly checked before recursing to $y$.
Up to $K$ subgraphs are collected in time $O(fK\alpha h)$.

\subsection{Well-Separated Pairs}
\label{sec:det-RPC_well-separated}

We classify the pairs $(P,F) \in \C$ with respect to their relevance
to a given node $x$ in a tree.
Recall that we want that the graph $G_x$ associated with $x$
contains no edge of $F$.
If so, the pair $(P,F)$ is \emph{active} for $x$.
Additionally, the number of edges of the path $P$ that are missing in $x$ should be ``small'',
where the exact amount depends on the depth of $x$ in the tree.
If so, the pair is \emph{well separated}.
If a pair is active for $x$, but still too many path edges are missing from $G_x$,
then the pair is \emph{poorly separated}.

\begin{definition}[well-separated, poorly-separated, active, and passive pairs]
\label{def:well-separated}
Let $i \in [K]$ and $x$ be a node of the tree $T_i$ at depth $d$.
Let further $(P,F)\in \C$ be a failure set-path pair.
\vspace*{.25em}
\begin{enumerate}
\item $(P,F)$ is \emph{well separated} with respect to $x$ if $F\subseteq A_x$ and $|P\cap A_x|\leq p^d L$.
\vspace*{.25em}
\item $(P,F)$ is \emph{poorly separated} with respect to $x$ if $F\subseteq A_x$ and $|P \cap A_x|> p^d L$.
\vspace*{.25em}
\item $(P,F)$ is \emph{passive} with respect to $x$ if $F\nsubseteq A_x$; otherwise, it is \emph{active} w.r.t.\ $x$.
\end{enumerate}
\vspace*{.25em}
We also say that the node $x$ is \emph{active} for $(P,F)$, if the pair $(P,F)$ is active w.r.t.\ $x$.
\end{definition}

We use the different categories of pairs as follows.
Let $(P,F) \in \mathcal{C}$ be a pair, where the path $P$ has endpoints $s,t \in V$. 
\autoref{alg:tree-exact} traverses a path from the root to a leaf in the tree $T_i$
by always selecting, at each node $x$, the first child $y$ of $x$ with $F \subseteq A_y$,
that is, the first child node that is \emph{active} for $(P,F)$.
We claim that, in order for this traversal to succeed for the query $(s,t,F)$ in $T_i$,
it is enough to maintain the following property.

\begin{center}
	\emph{If $y$ is the first child node of $x$ that is active for $(P,F)$,\\
	then $(P,F)$ is also well separated with respect to $y$.}
\end{center}

\noindent
To see this, recall that the probability parameter $p$ satisfies $p^h L < 1$.
Along the path traced by~\Cref{alg:tree-exact}, the fraction of edges of $P$
contained in the associated sets $A_x$ geometrically decreases.
In the leaf at depth $d = h$, we have both $F \subseteq A_x$ 
as well as $|P \cap A_x| \le p^h L < 1$
and thus $P \cap A_x = \emptyset$. 
In other words, $G_x = G{-}A_x$ contains no edge of $F$ but all of $P$.

To track the well and poorly-separated pairs, we define the following 
binary variables.
Let $x$ be a node in some tree $T_i$ and let $(P,F)\in \C$ be a pair that is active w.r.t.\ $x$. 

\vspace*{-1em}
\begin{align*}
X_{P,F,x}&=
\begin{cases}
1&~\text{if } (P,F) \text{ is well separated w.r.t. }x;\\
0&~\text{otherwise}.
\end{cases}\\[.5em]
Y_{P,F,x}&=
\begin{cases}
1&~\text{if } (P,F) \text{ is poorly separated w.r.t. }x;\\
0&~\text{otherwise}.
\end{cases}
\end{align*}

Let $r_1, r_2, \dots, r_K$ denote the respective roots (at depth $d=0$) of $T_1, T_2, \dots, T_K$.
At least for the $r_i$,
we do not have to worry about separation, namely, $X_{P,F,r_i} \equiv 1$.
The follows immediately from the definition of well-separated pairs
and the roots representing empty graphs.
This sets up a starting point both for the derandomization of the tree $T_i$
and later for the query algorithm.

\begin{observation}
For $i\in [K]$, all pairs in $\C$ are well separated with respect to the root of $T_i$.
\end{observation}

During the derandomization,
we have to keep track of the pairs in $\C$ that are covered by the trees we have already processed.

\begin{definition}[handled pairs]
For any index $i \in [K]$,
a pair $(P,F) \in \mathcal{C}$ is \emph{handled} by the tree $T_i$ if $(P,F)$ is well separated
with respect to every node in the root-to-leaf path in $T_i$ 
that is traced by~\autoref{alg:tree-exact} with query set $F$.
\end{definition}

\begin{definition}[the set $\D_x$]
\label{def:pairs_Dx}
Let $i \in [K]$ and $x$ be a node of the tree $T_i$.
The set $\D_x \subseteq \C$ contains those pairs $(P,F)$ that satisfy the following properties.
\vspace*{.25em}
\begin{enumerate}
	\item $(P,F)$ is not handled by any of the trees $T_1,\ldots,T_{i-1}$.
	\vspace*{.25em}
	\item The path in the tree $T_i$ that~\autoref{alg:tree-exact} traces on input $F$ 
		leads from the root $r_i$ to $x$.
	\vspace*{.25em}
	\item $(P,F)$ is a well-separated pair with respect to each node in $T_i$
		on the path from $r_i$ to $x$.
\end{enumerate}
\end{definition}

\subsection{Derandomization Algorithm}
\label{sec:det-RPC_derand}

\begin{algorithm}[t]
Let $T_1,\ldots,T_K$ be the collection of (random) sampling trees rooted at nodes $r_1,\ldots,r_K$\;
\For{$i=1$ {\bf to} $K$}
{
    $\D_{r_i}\gets \C \setminus \bigcup \{ \D_x \mid x\text{ is a leaf in }T_{i_0}\text{ for some }i_0<i\}$\;     
    Derandomize-Subtree$(T_i,r_i,\D_{r_i})$\;
}
\caption{Derandomization of the Sampling Trees $T_1,\ldots,T_K$.}
\label{Algorithm:Deterministic-Sampling-Tree}
\end{algorithm}
\begin{procedure}[t]
Let $y_1,\ldots,y_\alpha$ be the child nodes of $x$ in $T$\;
\For{$j=1$ {\bf to} $\alpha$}
{
    $\D\gets \{(P,F) \in \D_x \mid (P,F) \text{ is passive w.r.t.}\ y_{k} 
    	\text{ for every } k < j\}$\; \label{line:retain_passive_pairs}
    \vspace*{.25em}
    derandomize the set $A_{y_j}$ edge by edge so that the conditional expectations of
    $\sum_{(P,F)\in \D} \big(X_{P,F,y_j}-\frac{1}{2}Y_{P,F,y_j}\big)$
    is maximized in each step\;
    \vspace*{.25em}
    $\D_{y_j}\gets \{ (P,F) \in \D \mid (P,F) \text{ is well separated w.r.t.}\ y_j\}$\;
    	\label{line:remove_well_separated}
    \lIf{$y_j$ is not a leaf}
    {Derandomize-Subtree$(T,y_j, \D_{y_j})$}
}
\caption{Derandomize-Subtree($T,x,\D_x$)}
\label{Procedure:Derandomize-Ti-x}
\end{procedure}

\autoref{Algorithm:Deterministic-Sampling-Tree} outlines the recursive derandomization procedure
for the trees $T_1,\ldots,T_K$. 
The recursive procedure \emph{Derandomize-Subtree} considers a node $x$ in the tree $T_i$.
It goes through the children $y_1,\ldots,y_\alpha$
in the same (arbitrary) order in which they are processed at query time by~\autoref{alg:tree-exact}.
Suppose that the current child node is $y_j$,
meaning that $x$ as well as $y_1, \dots, y_{j-1}$ are already derandomized. 
It has access to the deterministic edge sets $A_x, A_{y_1}, \ldots, A_{y_{j-1}}$
and needs to compute $A_{y_j}$. 
Let $\D_x \subseteq \C$ be the set of pairs described in~\autoref{def:pairs_Dx} and consider the following subset.
\begin{align*}
\D= \{(P,F) \in \D_x \mid (P,F) \text{ is passive with respect to } y_{j_0} 
    	\text{ for some } j_0 < j\}.
\end{align*}
Intuitively, the pairs in $\D$ are those for which the query algorithm
starting in $r_i$ has not only reached node $x$, but will also not recurse
in any of its child nodes before checking $y_j$.

We aim to maximize the number of pairs in $\D$ that are well separated with respect to $y_j$,
while at the same time, we do not want too many pairs from $\D$ to become poorly separated with respect to $y_j$. 
It is crucial to strike a balance here, as the poorly-separated pairs may not be handled by $T_i$
and may thus unduly increase the number of trees required.
For each edge $e \in A_x$, we decide one by one whether to include it in $A_{y_j}$.
To make this decision, we compare the respective expected values of
$$\sum_{(P,F)\in \D} \left(X_{P,F,y_j}-\frac{1}{2}Y_{P,F,y_j}\right)$$
\emph{conditional} on the event $[e \in A_{y_j}]$ and on $[e \notin A_{y_j}]$.
The expectation is taken over the random process that just removes any edge 
$e' \in A_x{\setminus}\{e\}$
for which we have not made a decision yet independently with probability $p$.
Finally, we set $\D_{y_j}$ to be those pairs in $\D$ that are now well separated with respect to $y_j$.
Note that this indeed satisfies~\autoref{def:pairs_Dx} for node $y_i$.

\subsection{Correctness and Analysis}
\label{sec:det-RPC_correctness}

In order to prove that our approach indeed yields a deterministic 
$(L,f)$-replacement path covering, we need to assign values to the parameters.
We leave the total number $K$ of trees open until the end of this section,
and instead start by fixing the branching factor to $\alpha = {(2L)^\frac{f}{h}}$ as well as the
the sampling probability to $p=(2L)^{-\frac{1}{h}}$.
Both parameters depend on the height $h$, which will also be set later.
For now, it is enough to observe that the definitions ensure both $p^f = \frac{1}{\alpha}$ as well as $p^h L = \frac{1}{2} < 1$.

Let $x$ be a non-leaf node in any of the trees and let $y$ be one of its child nodes. 
Recall the random construction in which the edge set $A_y$ is obtained from $A_x$ by sampling any edge $e \in A_x$ independently with probability $p$. 
(See~\autoref{sec:det-RPC_prelims} for more details.)
The following lemma describes how this random construction behaves with respect to the passive, 
well-separated, and poorly-separated pairs.

\begin{lemma}
\label{lem:parent-to-child-well-separated-pair}
Let $i \in [K]$ be an index, $x$ a non-leaf node in the tree $T_i$,
$(P,F)\in \C$ a well-separated pair w.r.t.\ $x$,
and $y$ a child node of $x$,
Let $\Prob_1 = \Prob[ \nwspace (P,F) \textup{ is passive w.r.t.} \, y \nwspace ]$,
$\Prob_2 = \Prob[ \nwspace (P,F) \textup{ is well separated w.r.t.} \, y \nwspace ]$,
and $\Prob_2 = \Prob[ \nwspace (P,F) \textup{ is poorly separated w.r.t.} \, y \nwspace ]$.
\begin{multicols}{3}
\begin{enumerate}
	\item $\Prob_1 \le 1 - \frac{1}{\alpha}$;
	\vspace*{.25em}
	\item $\Prob_2 \ge \frac{1}{2}(1-\Prob_1)$;
	\vspace*{.25em}
	\item $\Prob_3 \le \frac{1}{2}(1-\Prob_1)$.
\end{enumerate}
\end{multicols}
\end{lemma}

\pagebreak

\begin{proof}
Let $d \in [h]$ be the depth of the child node $y$.
Since $(P,F)$ is well separated w.r.t.\ $x$, we have $F\subseteq A_x$ and $|P\cap A_x|\leq p^{d-1}L$.
Each element of $A_x$ lies in $A_y$ with probability $p$,
so the probability of $(P,F)$ being passive at $y$ 
is $\Prob_1 = \Prob[F \nsubseteq A_{y}] = 1-p^{|F|} \le 1-p^f = 1-\frac{1}{\alpha}$.
 
We now calculate the probability $\Prob_2$ of $(P,F)$ being a well-separated pair w.r.t.\ $y$,
meaning that both $F \,{\subseteq}\, A_y$ and $|P \cap A_y| \le p^d L$ hold.
Note that the two conditions are independent (see~\autoref{lem:tree-property})
and we have $\Prob[F \,{\subseteq}\, A_y] = (1-\Prob_1)$.
So we only need to estimate $\Prob\big[|P \cap A_y| \le p^d L \big]$.
Observe that $|P\cap A_y|$ follows the binomial distribution
$\operatorname{Bin}(|P\cap A_x|,p)$ and thus has median $p\cdot |P\cap A_x| \leq p^dL$.
Therefore, we get
\begin{equation*}
	\Prob\Big[ \nwspace |P \cap A_y| \le p^d L \Big] \ge \Prob\Big[\nwspace |P\cap A_y| \le p\cdot |P\cap A_x| \nwspace \Big] = \frac{1}{2}.
\end{equation*}
This gives
$\Prob_2 = \Prob\big[|P \cap A_y| \le p^d L \big] \cdot \Prob[F \,{\subseteq}\, A_y] 
	\ge \frac{1}{2} (1-\Prob_1)$.
The estimate for $\Prob_3$ follows immediately from $\Prob_1 + \Prob_2 + \Prob_3 = 1$.
\end{proof}

Finally, we can use the randomized construction to prove properties of its derandomization.

\begin{lemma}
\label{lem:derand-consecutive-layers}
Let $i \in [K]$ and $x$ a non-leaf node in the tree $T_i$
with children $y_1,\ldots,y_\alpha$. 
The sets $\D_{y_1},\ldots,\D_{y_\alpha}$ computed by \textup{Derandomize-Subtree}$(T_i,x,\D_x)$ 
satisfy $\sum_{j=1}^\alpha |\D_{y_j}|\geq  \frac{2}{11}|\D_x|$.   
\end{lemma}

\begin{proof}
Fix a child $y_j$. 
Let $\D= \{(P,F) \in \D_x \mid (P,F) \text{ is passive w.r.t.}\ y_{j_0} 
    	\text{ for some } j_0 < j\}$.
We process the derandomization of $A_{y_j}$ so that the value of 
$\sum_{(P,F)\in \D} \left(X_{P,F,y_j}-\frac{1}{2}Y_{P,F,y_j}\right)$ is at least its expectation.
Fix a pair $(P,F)\in \D$.
Due to $\D \subseteq \D_x$ this pair is well-separated w.r.t.\ the node $x$,
so we can apply~\autoref{lem:parent-to-child-well-separated-pair}.
Let $\Prob_1$, $\Prob_2$, and $\Prob_3$ be defined as above for the child node $y = y_i$.
We then have $\Exp[X_{P,F,y_j}] = \Prob_2 \ge \frac{1}{2}(1-\Prob_1)$
as well as $\Exp[Y_{P,F,y_j}] = \Prob_3 \le \frac{1}{2}(1-\Prob_1)$. 
This implies that the (expected and thus deterministic) value of that sum is at least
\begin{equation}
\label{eq:expected_value}
	\sum_{(P,F)\in \D} \left(X_{P,F,y_j}-\frac{1}{2} \cdot Y_{P,F,y_j}\right) 
		\ge |\D| \left( \frac{1}{2}(1-\Prob_1) - \frac{1}{4}(1-\Prob_1) \right) \ge \frac{|\D|}{4} (1-\Prob_1) \ge \frac{|\D|}{4\alpha}.
\end{equation}

For simplicity, assume an idealized scenario where no pair in $\D_x$
ever becomes poorly separated with respect to $y_j$, for any $j \in [\alpha]$.
On the one hand, this means $Y_{P,F,y_j} = 0$ 
and thus the estimate in (\ref{eq:expected_value})
states that at least $|\D|/4\alpha$ pairs become well separated in one iteration.
On the other hand, the only way to remain passive w.r.t.\ $y_j$
is to \emph{not} become well separated.
\autoref{line:retain_passive_pairs} and \ref{line:remove_well_separated} 
of the Derandomize-Subtree procedure then ensure that
$\D = \D_x{\setminus} \bigcup_{k = 1}^{j-1} \D_{y_{k}}$ holds in the $j$-th iteration.
Since the $\D_{y_k}$ are pairwise disjoint, this yields
	$|\D_{y_j}| \ge \frac{|\D_x| - \sum_{k=1}^{j-1} |\D_{y_k}|}{4\alpha}$.
Summing over all iterations gives
\begin{equation*}
	\sum_{j=1}^\alpha |\D_{y_j}| \ge \alpha \cdot \frac{|D_x|}{4\alpha} 
		- \sum_{j=1}^\alpha \sum_{k=1}^{j-1} \frac{|D_{y_k}|}{4\alpha} 
		= \frac{|D_x|}{4} - \frac{1}{4} \sum_{j=1}^\alpha (\alpha{-}j) \frac{|D_{y_j}|}{\alpha}
		\ge \frac{|D_x|}{4} - \frac{1}{4} \sum_{j=1}^\alpha |D_{y_j}|,
\end{equation*}
which is solved by $\sum_{j=1}^\alpha |\D_{y_j}| \ge |\D_x|/5$.

Now let $\mathcal{P}_j \subseteq \D_x$ be the pairs that become poorly separated w.r.t.\ $y_j$.
They do not contribute to $\D_{y_j}$ and also do not get included into $\D$
at the beginning of the next iteration.
This worsens the lower bound for the individual terms to
	$|\D_{y_j}| \ge \frac{|\D_x| - \sum_{k=1}^{j-1} |\D_{y_k}| 
		- \sum_{k=1}^{j-1} |\mathcal{P}_k|}{4\alpha}$.
However, the estimate (\ref{eq:expected_value}) being non-negative implies
that for any two pairs that become poorly separated with respect to $y_j$, 
at least one more pair must have become well separated, i.e., $|\D_{y_j}|/2 \ge |\mathcal{P}_j|$.
The same argument as above gives
$\sum_{j=1}^\alpha |\D_{y_j}| \geq 2|\D_x|/{11}$.
\end{proof}

As a direct corollary of the lemma, we get that the fraction of pairs not yet handled
by the trees $T_1,\ldots,T_{i-1}$ that get handled by $T_i$ is at least $(2/11)^h$.
Slightly abusing notation, let
$\D_{T_i} = \bigcup \nwspace \{ \D_x \mid x\text{ is a leaf in }T_{i} \}$.
The number of pairs handled by the $K$ trees is at least
\begin{equation*}
	\sum_{i=1}^K |\D_{T_i}| \ge |\C| \cdot \left(\frac{2}{11}\right)^h 
		\sum_{i=1}^K \left(1-\left(\frac{2}{11}\right)^h \right)^{i-1} 
		= |\C| \cdot \left( 1-\left(1-\left(\frac{2}{11}\right)^h \right)^K \right).
\end{equation*}
Since $\sum_{i=1}^K |\D_{T_i}|$ is integral,
choosing $K$ large enough so that the lower bound is strictly larger than $|\C|-1$
is sufficient for the sum to exhaust all $|\C|$ pairs.
This condition simplifies to $(1-(\frac{2}{11})^h)^K < \frac{1}{|\C|}$
and is satisfied, e.g., by
$K = 2 \nwspace (\frac{11}{2})^h \ln |\C| = O( (\frac{11}{2})^h f \log n)$
since
\begin{equation*}
	\left(1-\left(\frac{2}{11}\right)^h \right)^K  
		\le \exp\!\left( -\left(\frac{2}{11}\right)^h \cdot K \right) 
		= \exp(-2 \ln |\C|) = \frac{1}{|\C|^2} < \frac{1}{|\C|}.
\end{equation*}

We can now calculate the covering value and query time depending on $f$, $L$, $n$, and $h$.

\begin{itemize}
	\item The size of family $\G$ is $O(K\cdot L^f)=O(  (\frac{11}{2})^h \cdot fL^f \log n )$.
	\vspace*{.25em}
	\item The query time to compute $\G_F$ is 
		$O(f\alpha h K)=O((\frac{11}{2})^h \nwspace h \cdot f^2 L^{f/h} \log n)$.
	\vspace*{.25em}
	\item For any set $F$, the size of family $\G_F$ is at most $K=O((\frac{11}{2})^h \cdot f\log n)$.
\end{itemize}
We finally choose $h=\sqrt{f\log_2L}$.
For $f = o(\log L)$, this ensures
	$\left(\frac{11}{2} \right)^h = L^{\frac{\log_2\left(\frac{11}{2} \right)}{\log_2 L} \cdot h}
		= L^{\log_2\left(\frac{11}{2} \right) \sqrt{\frac{f}{\log_2 L}} } =  L^{o(1)}$
as well as $L^{\frac{f}{h}} = L^{ \sqrt{\frac{f}{\log_2 L}}} = L^{o(1)}$.
Finally, we obtain the parameters stated in~\autoref{thm:det-trees}.

\begin{itemize}
	\item The covering value is $O((\frac{11}{2})^h \cdot fL^f \log n ) = O(f L^{f+o(1)} \log n)$.
	\vspace*{.25em}
	\item The query time is 
		$O((\frac{11}{2})^h \nwspace h \cdot f^2 L^{f/h}\log n) = O(f^{5/2} L^{o(1)} \log n)$.
	\vspace*{.25em}
	\item The size of $\G_F$ is $O((\frac{11}{2})^h \cdot f\log n) = O(f L^{o(1)} \log n)$.
\end{itemize}

For higher sensitivities up to $f = o(\log n)$ the $L^{o(1)}$ terms turn to $n^{o(1)}$ since
	$\left(\frac{11}{2} \right)^h = n^{\frac{\log_2\left(\frac{11}{2} \right)}{\log_2 n} \cdot h}
		=  n^{\log_2\left(\frac{11}{2} \right) \cdot \frac{\sqrt{f \log_2 L}}{\log_2 n}} =  n^{o(1)}
	\ \text{and}\ 
	L^{\frac{f}{h}} = n^{\frac{\log_2 L}{\log_2 n} \cdot \frac{f}{h}}
		= n^{\frac{\sqrt{f \log_2 L}}{\log_2 n}} = n^{o(1)}$.

\section{Randomized Replacement Path Coverings}
\label{sec:rand-RPC}

We now improve the upper bound on the covering value of (randomized) $(L,f)$-replacement path coverings
for the parameter range $f = o(\log L)$.
We build on the sampling trees by Bilò et al.~\cite{Bilo25IndexingSubnetworks},
reviewed in~\autoref{sec:det-RPC_prelims}.
Like before, each of the $K$ trees has height $h$ and any internal node has exactly $\alpha$ children. We also reuse~\autoref{alg:tree-exact} to process queries.
The main difference is that we give a tighter analysis of the parameters of the construction
which allows us to derive better bounds.
In the following, we let $P(s,t,F)$ denote a \emph{replacement path} from $s$ to $t$,
that is, a shortest $s$-$t$-path in $G{-}F$.
If there are multiple such paths, we take one with the minimum number of edges.
Recall that for any node $x$ in a sampling tree, we use $G_x \subseteq G$
for the spanning subgraph associated with $x$.

\begin{lemma}[Lemma~7 in~\cite{Bilo25IndexingSubnetworks}]
\label{lem:parent-to-child}
	Let $i \in [K]$ be an index, $F \subseteq E$ a set of $|F| \le f$ edges,
	and $s,t \in V$ vertices such that $P(s,t,F)$ contains at most $L$ edges.
	\vspace*{.25em}
	\begin{enumerate}
		\item \autoref{alg:tree-exact} reaches a leaf of the sampling tree $T_i$ with probab.\
        	at least $((1-(1-p^f)^{\alpha})^h$.
        \vspace*{.25em}
		\item If~\autoref{alg:tree-exact} reaches a leaf $x$ of the tree $T_i$,
			then the probability of the path $P(s,t,F)$ 
            existing in $G_x$ is at least $(1-p^h)^L$.
	\end{enumerate}
\end{lemma}

The total number of nodes is $O(K \alpha^h)$
and the query algorithm selects $|\G_F| \le K$ leaf nodes in time $O(fK \alpha h)$. 
The probability to reach a leaf of a tree is exponentially small in $h$ 
(\autoref{lem:parent-to-child}.1).
However, once a leaf is reached, the probability of the stored graph holding
the relevant path $P(s,t,F)$ \emph{grows} with $h$, depending on $p$.
As before, we need to cover (with high probability) 
all pairs of vertices $s,t \in V$ for which $P(s,t,F)$ has at most $L$ edges.
Let $\delta > 0$ be a sufficiently large constant.
Diverging from~\cite{Bilo25IndexingSubnetworks}, we choose the parameters
\vspace*{-.75em}
\begin{multicols}{2}
  \begin{itemize}
	\item $h = \sqrt{f \ln (L/f)}$
	\vspace*{.5em}
	\item $K = \delta \nwspace e^f (\tfrac{e}{e-1})^h  f \ln n$
	\item $\alpha = (L/f)^{f/h}$
	\item $p = (f/L)^{1/h}$
\end{itemize}
\end{multicols}
\vspace*{-.5em}

\autoref{lem:tree-property} states that in any leaf $x$, at depth $h$,
the probability for any edge to be removed (event $[e \in A_x]$) is $p^h = f/L$.
We verify next that the query algorithm indeed computes a suitable collection of graphs.
The lemma also implies that the graphs stored in the leaves of the trees
form an $(L,f)$-replacement path covering with high probability.

\begin{lemma}
\label{lem:correctness}
	W.h.p.\ over all sets $F \subseteq E$ with $|F| \le f$ and $s,t \in V$ with $|E(P(s,t,F))| \le L$,
	upon termination of~\autoref{alg:tree-exact},
	there exists a graph $G_x \in \G_F$ that retains $P(s,t,F)$.
\end{lemma}

\begin{proof}
	By~\autoref{lem:parent-to-child},
	the probability to reach a leaf $x$ in some tree $T_i$ whose corresponding graph $G_x$
	contains the replacement path $P(s,t,F)$ is at least $((1-(1-p^f)^{\alpha})^h \cdot (1-p^h)^L$.
	We insert the parameters into the first factor and apply a lower bound of the form
	$(1 {+} \frac{t}{k})^k \ge (1 {-}\frac{t^2}{k}) \nwspace e^{t}$ for $k \ge 1$ and $k \ge |t|$.
	It can, for example, be found in the textbook
	by Motwani and Raghavan~\cite[p.~435]{Motwani95RandomizedAlgorithms}.
	Recall that we assume $f = o(\log L)$; for now, it will be sufficient that $2f^2 \le L$.
	Regarding the second factor, we have
		$(1-p^h)^L = \left(1-\frac{f}{L} \right)^L 
			\ge \left(1- \frac{f^2}{L} \right) \cdot \frac{1}{e^f}
			\ge \frac{1}{2e^f}$
	Further, observe that 
		$\left( 1- \left(1-p^f \right)^{\alpha} \right)^h 
			= \left(1-\left(1-\frac{1}{(L/f)^{f/h}} \right)^{(L/f)^{f/h}} \right)^h 
			\ge \left(1- \frac{1}{e} \right)^h$.
	Let $C$ abbreviate $\frac{e}{e-1} \approx 1.582$. 
	We have thus shown that the probability is at least
	$1/(2e^f \nwspace C^h)$.
	
	Repeating the query in $K = \delta \nwspace e^f C^h f \ln n$
	independent trees reduces the failure probability for any triple $(s,t,F)$ to
		$\left(1- \frac{1}{2 e^f \nwspace C^h} \right)^{\! \delta \nwspace  e^f C^h f \ln n}
			\!\!\!\le \left(\frac{1}{e}\right)^{\!\frac{\delta}{2}  \ln n} 
			\!\!= n^{-\frac{\delta f}{2}}$.
	There are at most $|V^2 \times \binom{E}{\le f}| = O(n^{2+2f})$ relevant triples $(s,t,F)$,
	which shows that choosing $\delta$ large enough and taking a union bound over all triples
	ensures a high success probability.
\end{proof}

Note that
	$C^h = C^{\sqrt{f \ln (L/f)}} 
    		= \Big(\left(\frac{L}{f}\right)^{\frac{\ln C}{\ln (L/f)}} \Big)^{\sqrt{f \ln (L/f)}}
	    = \Big(\frac{L}{f}\Big)^{\ln C \sqrt{\frac{f}{\ln (L/f)}}} 
	    = \Big(\frac{L}{f}\Big)^{o(1)}$.
The last estimate uses that $f = o(\log L)$ also gives $f = o(\ln (L/f))$. 
Similarly, we have 
\begin{equation*}
    \alpha = \Big(\frac{L}{f}\Big)^{f/h} = \Big(\frac{L}{f}\Big)^{\sqrt{f/\ln (L/f)}} = \Big(\frac{L}{f}\Big)^{o(1)}.
\end{equation*}
Our choice of parameters thus shows that the whole data structure stores 
$$O(K \alpha^h) = O(f e^f C^h (L/f)^f \ln n) = O(f e^f (L/f)^{f+o(1)} \ln n)$$
graphs and has a query time of
$O(fKh\alpha) = O(f^\frac{5}{2} e^f(L/f)^{o(1)} \sqrt{\ln (L/f)}  \ln n)$.

The size of the subfamily $\G_F$ is $K = O(fe^f(L/f)^{o(1)} \ln n)$ in the worst case.
However, with high probability the number of subgraphs relevant to a query $F$ is smaller.
In~\autoref{lem:correctness}, we showed that in any tree $T_i$ the probability of finding a leaf node $x$ satisfying $F \subseteq A_x$ is $(1{-}\frac{1}{e})^h = 1/C^h$. Thus, the expected size of ${\cal G}_F$ is $K/C^h=O(fe^f\log n)$. Further recall that each $T_i$ adds a graph to ${\cal G}_F$ independently of other trees. 
Applying a standard Chernoﬀ bound
shows that also with high probability $\G_F$ contains $O(fe^f\log n)$ graphs.

\section{Lower Bound} 
\label{sec:lower}

We now establish the lower bound on the covering value.
Slightly abusing notation, 
we show that there exists families of graphs for which any $(L{+}1,f)$-replacement path covering
must contain at least $\sum_{i=0}^{f-1} \binom{L-1}{i}$ subgraphs
to fulfill the requirements of~\autoref{def:RPC}.
That means, we set denote the cut-off parameter as $L+1$ (instead of $L$) to ease the notation below.
Our construction is assembled from building blocks which we call \emph{inner trees}.
These are comb-like structures which are connected by iteratively replacing the leafs 
of one inner tree with new, smaller inner trees.
The sizes are arranged in such a way that in the emerging binary tree 
all leaves have the same distance $L$ from the root.

\subsection{Inner Trees}
\label{subsec:lower_construction}

We fix a sensitivity $f$ and cutoff-value $L+1$ for the rest of this subsection.
For a positive integer $k$, an \emph{$(L,k)$-inner tree} is constructed 
from a directed path $P = (y_{L}, \ldots, y_2, y_1)$ with $L-1$ edges. 
The path is rooted at node $y_{L}$. For each $i \in [L]$, we add a directed edge from $y_i$
to a new leaf node $z_i$.
The weight of the edge $(y_i, z_i)$ is set to $i \cdot k$.
The edges in $P$ all have weight zero.
See~\autoref{fig:enter-label} for an illustration.

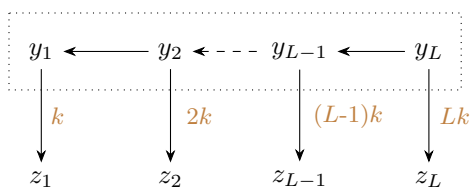
\begin{figure}[t]
\centering~
\begin{tikzpicture}[
    scale=0.85,
    every node/.style={draw=none, minimum size=0cm, anchor=center}
]

\node (yL) at (0,0) {$y_L$};
\node (y_d1) at (-2,0) {$y_{L-1}$};
\node (y2) at (-4,0) {$y_2$};
\node (y1) at (-6,0) {$y_1$};

\node (zL) at (0,-2) {$z_L$};
\node (z_d1) at (-2,-2) {$z_{L-1}$};
\node (z2) at (-4,-2) {$z_2$};
\node (z1) at (-6,-2) {$z_1$};

\draw[<-,>=Stealth] (y_d1) -- (yL);
\draw[<-,>=Stealth,dashed] (y2) -- (y_d1);
\draw[<-,>=Stealth] (y1) -- (y2);

\draw[->,>=Stealth] (y1) -- (z1);
\draw[->,>=Stealth] (y2) -- (z2);
\draw[->,>=Stealth] (yL) -- (zL);
\draw[->,>=Stealth] (y_d1) -- (z_d1);


\node at (-5.75,-1) {\color{brown}\small $k$};
\node at (-3.55,-1) {\color{brown}\small $2k$};
\node at (-1.25,-1) {\color{brown}\small 
$(L$-$1)k$};
\node at (0.4,-1) {\color{brown}\small $Lk$};

\draw[dotted] (-6.5,0.6) rectangle (0.5,-0.6);

\end{tikzpicture}
\caption{Depiction of an $(L,k)$-inner tree}
\label{fig:enter-label}
\end{figure}

We construct a binary tree $T$ in $f$ rounds, numbered from $0$ to $f-1$.
In the $0$-th round, we initialize $T$ as an $(L, L^{2f})$-inner tree.
Its root serves as the root of the whole tree $T$.
In each subsequent round $j$, 
we replace each leaf in $T$ at hop-distance $d$ from the root with an $(L{-}d, L^{2(f-j)})$-inner tree. Note that at the end of round $f-1$,
every root-to-leaf path in the resulting tree has $L$ edges.
See~\autoref{fig:tree-T}.

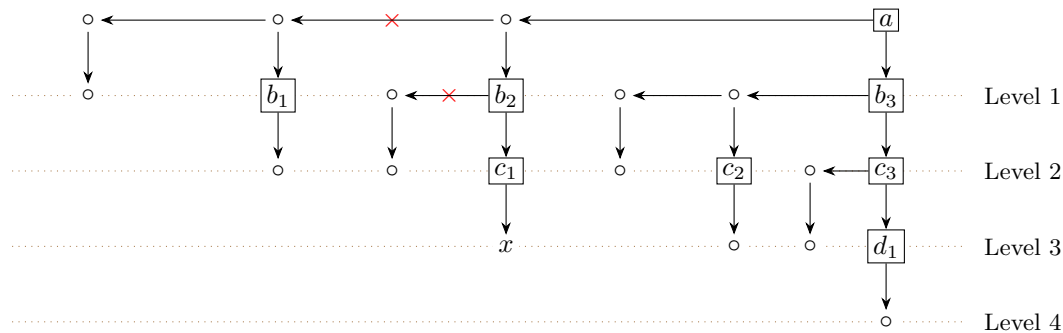
\begin{figure}[b]
\centering
\begin{tikzpicture}[every node/.style={draw=none, minimum size=0cm, fill=white, inner sep=2pt}]

    \node (y1) at (0.5, 0) {$\circ$};
    \node (y2) at (3, 0) {$\circ$};
    \node (y3) at (6, 0) {$\circ$};
    \node[draw] (y4) at (11, 0) {$a$};

\draw[dotted,gray!30!brown] (-0.5, -1) -- (12, -1);
\draw[dotted,gray!30!brown] (-0.5, -2) -- (12, -2);
\draw[dotted,gray!30!brown] (-0.5, -3) -- (12, -3);
\draw[dotted,gray!30!brown] (-0.5, -4) -- (12, -4);
\node(L1) at (12.8, -1) {\small Level $1$};
\node(L2) at (12.8, -2) {\small Level $2$};
\node(L3) at (12.8, -3) {\small Level $3$};
\node(L4) at (12.8, -4) {\small Level $4$};

    \node (z0) at (0.5, -1) {$\circ$};
    \node[draw] (z1) at (3, -1) {$b_1$};
    \node (z2) at (4.5, -1) {$\circ$};
    \node[draw] (z3) at (6, -1) {$b_2$};
    \node (z4) at (7.5, -1) {$\circ$};
    \node (z5) at (9, -1) {$\circ$};
    \node[draw] (z6) at (11, -1) {$b_3$};

    \node (x1) at (3, -2) {$\circ$};
    \node (x2) at (4.5, -2) {$\circ$};
    \node[draw] (x3) at (6, -2) {$c_1$};
    \node (x4) at (7.5, -2) {$\circ$};
    \node (x5) at (7.5, -2) {$\circ$};
    \node[draw] (x6) at (9, -2) {$c_2$};
    \node (x7) at (10, -2) {$\circ$};
    \node[draw] (x8) at (11, -2) {$c_3$};

    \node (w1) at (6, -3) {$x$};
    \node (w2) at (9, -3) {$\circ$};
    \node (w3) at (10, -3) {$\circ$};
    \node[draw] (w4) at (11, -3) {$d_1$};
    
    \node (a1) at (11, -4) {$\circ$};

\node(X0) at (4.5, 0) {\color{red} $\times$};
\node(X1) at (5.25, -1) {\color{red} $\times$};

    \draw[->,>=Stealth] (y4) -- (y3);
    \draw[->,>=Stealth] (y2) -- (y1);
    \draw[->,>=Stealth] (z6) -- (z5);
    \draw[->,>=Stealth]  (z5) -- (z4); 
    \draw[->,>=Stealth] (y3) -- (y2);
    \draw[->,>=Stealth] (z3) -- (z2);
    \draw[->,>=Stealth] (x8) -- (x7);
    
    \draw[->,>=Stealth] (y4) -- (z6);
    \draw[->,>=Stealth] (z6) -- (x8);
    \draw[->,>=Stealth] (x8) -- (w4);
    \draw[->,>=Stealth] (x7) -- (w3);
    \draw[->,>=Stealth] (z5) -- (x6);
    \draw[->,>=Stealth] (x6) -- (w2);
    \draw[->,>=Stealth] (z4) -- (x5);
    \draw[->,>=Stealth] (z1) -- (x1);
    \draw[->,>=Stealth] (z2) -- (x2);
    \draw[->,>=Stealth] (z3) -- (x3);
    \draw[->,>=Stealth] (y2) -- (z1);
    \draw[->,>=Stealth] (y3) -- (z3);
    \draw[->,>=Stealth] (x3) -- (w1);
    \draw[->,>=Stealth] (y1) -- (z0);
    \draw[->,>=Stealth] (w4) -- (a1);    
\end{tikzpicture}
\caption{Depiction of the tree $T$ for parameters $L=4$ and $f=3$. 
The edges marked with a cross are the set $F_x$ for the leaf node $x$ at level $3$.}
\label{fig:tree-T}
\end{figure}

For any $j \in [f]$, let $N(L,j)$ denote the number of leaf nodes in the tree $T$ at level~$j$.
Let further $R(L,j)$ be the number of $j$-tuples $(r_1,\ldots,r_j)$
of positive integers $r_1,\ldots,r_j\geq 1$ that sum to exactly $L$. 
Then, it holds that $R(L,j) = \binom{L-1}{j-1}$.
Observe that $N(L,j)\geq R(L,j)$. This is because each distinct tuple $(r_1,\ldots,r_j)$
maps to a unique leaf of tree $T$ at level $j$.
The mapping is obtained by traversing $r_1$ hops in the first inner tree, 
$r_2$-hop-length in the second inner tree, and so on.
Therefore, $N(L,j) \geq \binom{L-1}{j-1}$. This shows that the total number of leaf nodes in tree $T$ is lower bounded from below by $\sum_{j=1}^f \binom{L-1}{j-1}$.

Let $s$ denote the root of $T$ and $X$ the set of leaves.
We extend $T$ to a digraph $G$ by taking a new sink $v$ and adding an edge
from each leaf node in $X$ to $v$ of weight $1$.
We demonstrate that for any $(L{+}1,f)$-replacement path covering for $G$ 
the size of the family $\G$ is at least $|X|$.

Fix some $x\in X$, we define a set $F_x \subseteq E(G)$ of at most $f$ edges as follows.
Let $\tau_{1},\ldots,\tau_{\ell}$ be the inner trees that are traversed
by the path from $s$ to $x$ in the tree $T$. 
For each $i \in [\ell]$, let $y_{i}$ be the ancestor of $x$
that lies in $\tau_{i}$ and is furthest from $s$.
The set $F_x$ is then formed by taking those out-edges of the nodes $y_{1}, \ldots, y_{\ell}$
that do \emph{not} lie on the unique path from $s$ to $x$ in $T$.
Let $T[s, x]$ denote this path.
Clearly, we have $|F_x| \le \ell \le f$.
Moreover, the edge weights are so that traveling inside one level is free.
It is always preferred to descend to the next level as far from the 
root of the current inner tree as possible.
Therefore, after the failures, $T[s,x]$ is the cheapest way to reach a leaf of $T$,
whence the concatenation $T[s, x] \circ (x,v)$
is the unique shortest path from $s$ to $v$ in the graph $G{-}F_x$.
The path has $L+1$ edges.

For each leaf $x\in X$ of $T$, let $G_x\in \G$ be a graph that avoids the edges in $F_x$ but still contains the replacement path $T[s, x] \circ (x,v)$. 
For any two distinct $x, x' \in X$, we assert that $G_x \neq G_{x'}$. This is due to the fact that if $G_x$ and $G_{x'}$ were the same, both would include the out-edges of the lowest common ancestor $w = \operatorname{LCA}_T(x, x')$. 
However, this is a contradiction since at least one of the out-edges of $w$ 
is included in $F_x$ or $F_{x'}$.
From this, we get our lower bound on the covering value $|\G| \ge |X| \ge \sum_{j=1}^f \binom{L-1}{j-1}$,
proving~\autoref{thm:lower-bound}.

\subsection{Proof of \texorpdfstring{\autoref{cor:lb}}{Corollary~4}}
\label{sec:lower_corollary}

When adjusting the cut-off value back to $L$ (from $L{+}1$) and parameterizing the binomial coefficient
in terms of $0 \le i \le f{-}1$ instead of $j \in [f]$, we get the sum $\sum_{i=0}^{f-1} \binom{L-2}{i}$.
We give a closed form for it below to prove~\autoref{cor:lb}.

\corlb*

The second clause is immediate from the observation that $f \ge \frac{L}{2}$ 
(that is, $f{-}1 \ge \frac{L-2}{2}$)
implies that summing the first $f-1$ binomial coefficients gives at least half the value of the sum over all $L-2$ terms.
\begin{equation*}
	\sum_{i=0}^{f-1} \binom{L-2}{i} \ge \frac{1}{2} \cdot \sum_{i=0}^{L-2} \binom{L-2}{i} = \frac{1}{2} \cdot 2^{L-2} = \Omega(2^L).
\end{equation*}

In the rest of the section, we prove the first clause.
If $f$ is bounded away from $L/2$,
meaning $f \le (1{-}\varepsilon) \frac{L}{2}$ for some positive constant $\varepsilon >0$,
the binomial coefficient $\binom{L-2}{f-1}$ is the largest term of the sum.
We bound that term using Stirling's approximation.
The following lemma can, for example, be found in the textbook by Cover and Thomas~\cite{Cover06InformationTheory}.

\begin{lemma}[Stirling's approximation]
\label{lem:Stirling}
	Let $k$ be a positive integer and $0 < x < 1$ a rational number such that $xk$ is integral.
	Then, it holds that
	\begin{equation*}
		\binom{k}{xk} \ge \frac{1}{\sqrt{8k \nwspace x(1{-}x)}} 
			\left(\frac{1}{x}\right)^{xk} \left(\frac{1}{1-x}\right)^{(1-x)k}.
	\end{equation*}
\end{lemma}

We apply \autoref{lem:Stirling} with $k = L-2$ and $xk = f-1$,
where we use that $f\ge 2$ and thus $x = (f{-}1)/(L{-}2) > 0$.
We lower bound the three factors on the right-hand side separately.
\begin{equation*}
	\left( \frac{1}{x}\right)^{xk} = \left( \frac{L-2}{f-1}\right)^{f-1} 
		= \left( \frac{L-2}{L}\right)^{f-1}	\left( \frac{L}{f-1}\right)^{f-1}.
\end{equation*}
It holds that $( \tfrac{L-2}{L})^{f-1} \ge  (\tfrac{L-2}{L})^{\frac{L}{2}} \ge 
(\tfrac{3}{4})^{2}$.
For the last estimate, we use that $(\tfrac{L-2}{L})^{\frac{L}{2}}$ is increasing\\[.25em]
in $L$ and that $L \ge 4$ due to $\frac{L}{2}(1-\varepsilon) \ge f \ge 2$.
This proves $x^{-xk} = \Omega \big( ( \tfrac{L}{f-1} )^{f-1} \big)$.

Next, observe that for $x < 1$, we have $(1-x)^{-1} \ge e^x$,
which can be derived from the arguably more common inequality $1-x \le e^{-x}$.
Since $f \le L/2$, we get
	$1-x = \frac{L-f-1}{L - 2} \ge \frac{1}{2}$.
Therefore, it holds that
	$\left(\frac{1}{1-x}\right)^{(1-x)k} \ge \exp(x \cdot (1{-}x)k) 
		\ge \exp\!\left(\! \frac{f-1}{2} \!\right) = \Omega\big(\sqrt{e^{f}} \nwspace \big)$.
Finally, we have 
	$(8k \nwspace x(1{-}x))^{-1/2} \ge (4 xk)^{-1/2} 
		= \Omega\!\left(\frac{1}{\sqrt{f-1}}\right)$.

Combining the three estimates results in
\begin{multline*}
	\binom{L-2}{f-1} = 
		\Omega\!\left( \frac{1}{\sqrt{f{-}1}} \left( \frac{L}{f{-}1} \right)^{f-1} 
			\sqrt{e^{f}} \right)\\
		 = \Omega\!\left( \sqrt{(f{-}1) \, e^{f}} \cdot  \frac{L^{f-1}}{(f{-}1)^f} \right)
		 = \Omega\!\left( \sqrt{f e^{f}} \cdot  \frac{L^{f-1}}{f^f} \right).
\end{multline*}

\bibliography{ref}

\end{document}